\documentclass[sigconf,10pt]{acmart}

\usepackage{todonotes}
\usepackage{graphicx}
\usepackage{amsthm}
\usepackage[toc,page]{appendix}
\usepackage{subcaption}
\usepackage{array}
\usepackage{enumitem}
\usepackage{adjustbox}
\usepackage[ruled,vlined]{algorithm2e}
\newcommand{\name}{\textsc{Moonshine}}
\usepackage{xcolor}
\usepackage{soul}

\newcolumntype{R}[2]{%
    >{\adjustbox{angle=#1,lap=\width-(#2)}\bgroup}%
    l%
    <{\egroup}%
}
\newcolumntype{C}[1]{>{\centering\arraybackslash}p{#1}}

\newcommand*\rot{\multicolumn{1}{R{25}{1em}}}

\theoremstyle{definition}
\newtheorem{definition}{Definition}[section]
\theoremstyle{theorem}
\newtheorem{theorem}{Theorem}[section]
\theoremstyle{lemma}
\newtheorem{lemma}{Lemma}[section]

\makeatletter
\newcommand{\removelatexerror}{\let\@latex@error\@gobble}
\makeatother

\title{{\name}: An Online Randomness Distiller for Zero-Involvement Authentication}

\author{Jack West}
\affiliation{Loyola University Chicago}
\email{jwest1@luc.edu}

\author{Kyuin Lee}
\affiliation{University of Wisconsin--Madison}
\email{kyuin.lee@wisc.edu}

\author{Suman Banerjee}
\affiliation{University of Wisconsin--Madison}
\email{suman@cs.wisc.edu}

\author{Younghyun Kim}
\affiliation{University of Wisconsin--Madison}
\email{younghyun.kim@wisc.edu}
\orcid{0000-0002-5287-9235}

\author{George K. Thiruvathukal}
\affiliation{Loyola University Chicago}
\email{gkt@cs.luc.edu}

\author{Neil Klingensmith}
\affiliation{Loyola University Chicago}
\email{neil@cs.luc.edu}

\begin{abstract}
Context-based authentication is a method for transparently validating another device's legitimacy to join a network based on location.
Devices can pair with one another by continuously harvesting environmental noise to generate a random key with no user involvement.
However, there are gaps in our understanding of the theoretical limitations of environmental noise harvesting, making it difficult for researchers to build efficient algorithms for sampling environmental noise and distilling keys from that noise.
This work explores the information-theoretic capacity of context-based authentication mechanisms to generate random bit strings from environmental noise sources with known properties.
Using only mild assumptions about the source process's characteristics, we demonstrate that commonly-used bit extraction algorithms extract only about 10\% of the available randomness from a source noise process.
We present an efficient algorithm to improve the quality of keys generated by context-based methods and evaluate it on real key extraction hardware.
{\name} is a randomness distiller which is more efficient at extracting bits from an environmental entropy source than existing methods.
Our techniques nearly double the quality of keys as measured by the NIST test suite, producing keys that can be used in real-world authentication scenarios.
\end{abstract}


\copyrightyear{2021}
\acmYear{2021}
\setcopyright{acmlicensed}\acmConference[IPSN '21]{The 20th International Conference on Information Processing in Sensor Networks (co-located with CPS-IoT Week 2021)}{May 18--21, 2021}{Nashville, TN, USA}
\acmBooktitle{The 20th International Conference on Information Processing in Sensor Networks (co-located with CPS-IoT Week 2021) (IPSN '21), May 18--21, 2021, Nashville, TN, USA}
\acmPrice{15.00}
\acmDOI{10.1145/3412382.3458899}
\acmISBN{978-1-4503-8098-0/21/05}


\begin{document}

\maketitle

\section{Introduction}

\emph{Context-based authentication} is emerging as a solution to enable fast and convenient device authentication.
It verifies two devices' coexistence by comparing an independently generated random key or pin from an ambient source of randomness, such as wireless signal strength, acoustic noise, electrical noise, etc.

A pair of devices performing context-based authentication begin by individually harvesting noise from a shared source of randomness (Fig. \ref{fig:pipeline}).
Each device samples the environmental signal with an analog-to-digital converter and divides its sequence of samples into blocks.
Each block of the samples is then converted into a single bit by way of a bit extraction technique.
If the random source has a significant fraction of common-mode noise shared between the two authenticating devices, then the bit sequences they extract are substantially similar.
The underlying assumption of context-based authentication is that only devices which are physically near one another share enough common-mode noise to authenticate.
Distant devices, which are assumed not to be legitimate authenticators, generate significantly different bit sequences from the contextual information that they can observe, so they cannot authenticate with a legitimate device.

The environment around us is full of noise sources, but all noise sources ultimately have the same problem: we need some filter or transformation to translate samples of the noise source into the keyspace elements.
The question we are trying to answer here is how do we efficiently build transformation functions that can produce high-quality keys.

As we will demonstrate, most environmental noise sources generate signals with a low \emph{randomness density}.
Informally, we say that if we harvest a $k$-bit key from an environmental source, that key could be represented with a shorter sequence of $k-n$ bits.
More formally, we say that the Shannon entropy rate of a $k$-bit key generated from an environmental noise source is $H_e(\mathcal{X}) < H_U(\mathcal{X})$, where $H_U(\mathcal{X}) = k$ is the entropy rate of an iid sequence of $k$ uniformly-distributed Bernoulli random variables.

One possible solution to this problem would be to use the bit sequence extracted from the environmental noise source as a seed for a pseudorandom number generator (PRNG), which would produce keys that are indistinguishable from random.
Since PRNGs are deterministic, all devices that start with the same seed would produce the same key.
Although, if we assume that an illegitimate user who wishes to gain access to a network that is secured by context-based authentication knows how the PRNG converts seeds to keys---an assumption we must make---then the keys produced by the PRNG can be of no higher quality than the seeds used to generate them.
Using cryptographic hash functions to randomize a low-entropy bit sequence results in a similar problem.

In this work, we answer two open questions regarding the information theoretic properties of context-based key generation.
First, we ask \emph{what is the maximum amount of randomness we can extract (in bits per second) from some environmental source process given only its power spectral density?}
Second, we ask \emph{how can we increase the randomness density in keys generated by context-based methods?}
In the existing literature, no technique can repeatedly generate shared keys from environmental noise that can pass standard tests for randomness---Table \ref{tab:nist} shows a summary of NIST results for many recent pieces of work that report their results.
Without generating sufficiently random keys, we cannot be sure that we are excluding unauthorized users.

\begin{figure*}
    \centering
    \includegraphics[width=\hsize]{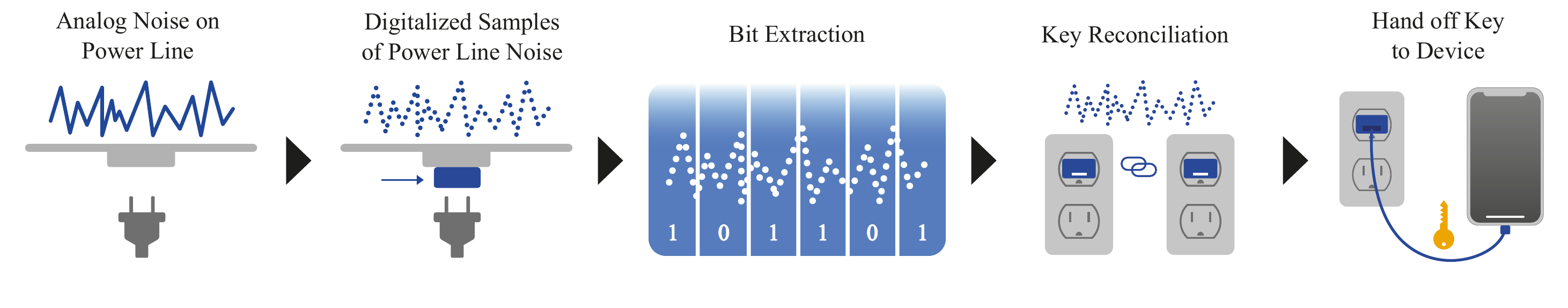}
    \vskip -20pt
    \caption{Data pipeline for \textsc{VoltKey}. {\name} is applied during the bit extraction phase.}
    \vskip -5pt
    \label{fig:pipeline}
\end{figure*}

To address the first problem, we develop a method for estimating the entropy rate---which is an upper bound on the bit extraction rate---of the raw environmental noise signal.
The bit extraction rate matters because the most secure keys are long.
RSA and DSA keys are between 2,048--4,096 bits in length, and the minimum length is getting longer all the time due to improvements in computational capacity available for brute force attacks.
For context-based authentication techniques to be practical, they must be faster than manual techniques like typing in a password.
But current context-based authentication schemes are slow for two reasons.
First, the entropy rate of the environmental noise process---measured in bits per second---limits the speed.
Second, the efficiency of standard key generation algorithms that extract bits from the chosen environmental noise processes tends to be low.
Standard algorithms can extract one key bit every 10--20 samples of the noise process.


To address the second problem, we introduce a new randomness distiller called {\name}.
{\name} distills randomness from a long bit sequence with low-entropy density into a shorter bit sequence with high entropy density that can be used as a secure cryptographic key.
{\name} produces keys with near-optimal entropy rate by selectively discarding subsequences from the input.
By contrast, other randomness correctors have suboptimal performance~\cite{drng}.

The technical challenges of implementing a randomness distiller stem from the fact that random numbers' predictability is difficult to measure.
There are many senses in which a bit sequence may be predictable: its periodicity, propensity to generate a particular bit sequence (even if that sequence's appearance is not periodic), etc.
Furthermore, once we know that a bit sequence is predictable, it is not easy to selectively eliminate its predictable elements.
Our theoretical analysis of environmental entropy sources suggests techniques to deal with these challenges and motivates our design of \textsc{Moonshine}.

This work introduces a novel randomness distiller in the data pipeline shown in Fig. \ref{fig:pipeline} between bit extraction and key reconciliation which selectively removes bits generated by the bit extractor to improve the quality of the final key that is generated.
Our randomness distiller introduces the new technique of discarding blocks of bits from the input sequence in a periodic fashion.
This has the effect of disrupting the periodic repetitions on the input sequence, with the result being more random sequences that do well on the NIST benchmarks.
We demonstrate that this bit discarding technique significantly improves the overall quality of keys while adding only a negligible processing overhead.

We demonstrate the effectiveness of {\name} by evaluating it on real context-based authentication hardware.
The authors of \textsc{VoltKey} lent us prototypes that we used to evaluate {\name}.
The {\name} corrector improves the randomness of bit sequences generated by \textsc{VoltKey} and considerably outperforms the Von Neumann randomness corrector, both in terms of the randomness of the resulting key and the amount of data retained after correction.
Keys generated by {\name} can pass 14/15 NIST tests for randomness, making them suitable for use as cryptographic and authentication keys.
It is important to know the source process's entropy rate because that imposes a hard upper limit on the speed of key extraction.
We might choose to work with one noise process or another, depending on the entropy rate.

Generally speaking, when we are designing a context-based authentication scheme, we know the general properties of the source process's power spectral density.
For example, body-area networks that harvest electrical signals from the heart (H2H~\cite{h2h} and H2B~\cite{Lin-H2B}) to generate authentication keys generally start with an analog source process that has a bandwidth of a few Hertz.
We can assume that its power spectral density will have a few strong harmonics in the 1--10 Hz frequency band superimposed on additive white Gaussian noise.
\textsc{VoltKey}, which harvests randomness from the electric power lines, will likely have strong harmonics at multiples of 60 Hz, tapering off at a few hundred Hertz.
Starting with this information, we show how to estimate the entropy rate of a typical realization of the source process and improve the quality of extracted keys.
\\

\noindent
The contributions of this work are as follows:
\begin{enumerate}
    \item We present a new method for calculating the entropy rate of a random process from its samples or its PSD \emph{that works for any wide-sense stationary random process} (\S \ref{sec:entropyrate} \& \S \ref{sec:fancyentropyrate}).
    \item Using intermediate results from the previous  technique, we introduce {\name}, \emph{an entropy distiller that can achieve a substantially improved pass rate of the NIST test for randomness}.
    \item We preformed a comparative analysis of key quality from different bit streams with different bit extraction algorithms and environmental sources.
    \item We build a prototype implementation of {\name} that runs on real context-based authentication hardware. We evaluate our prototype on several types of environmental noise. Our results show that {\name} is generalizable to different context based authentication mechanisms.
\end{enumerate}

Our methods make some mild assumptions about the characteristics of the environmental source noise process.
We model the source process as the sum of a band-limited deterministic signal and additive white Gaussian noise.

\section{Background}

Devices that use context-based security take advantage of the fact that the common contextual information is shared only by a limited group of closely located devices. 
The presence of common contextual information is evidence that the devices are located in the same place simultaneously, which implies that they legitimately belong to the same user.
The keys generated from contextual information can establish initial trust (as a pairing key) and protect subsequent communication (as a cryptographic key).
This eliminates the need for human involvement for making, entering, and managing a secret key, which can dramatically improve the overall usability of systems that currently rely on passwords to protect data.
In addition, the time-varying nature of contextual information also allows devices to use a new key for each pairing attempt or periodically update the cryptographic key, which significantly reduces the attack window for adversarial agents.

Similar studies \cite{Lin-H2B, proximate, Mayrhofer-Pervcomp07, floppuf}, show that the act of calculating the amount of randomness in an environmental signal is not straightforward.
We cannot just sample the source and create a histogram to compute the statistical entropy because the samples are not independent.
Most of the time, environmental noise contains some deterministic component, causing samples to be correlated.
This deterministic component makes the entropy computation meaningless.
Techniques exist for computing entropy rates from a signal's power spectral density (PSD), but they are unstable for correlated random processes.
Our contribution is to introduce a stable method to calculate the entropy rate from a signal's PSD that works on arbitrary random processes.

\subsection{An Overview of the Standard Context Based Authentication Process}
\label{sec:contextoverview}

This section gives a broad overview of the process devices and goes through to generate random keys from environmental noise.
We assume that there are two devices that are both located near each other and measuring the same random environmental noise.
Most context-based authentication schemes involve three basic steps: noise harvesting, key generation, and reconciliation.

\begin{enumerate}[leftmargin=0cm,itemindent=.5cm,labelwidth=\itemindent,labelsep=0cm,align=left]
    \item \textbf{Noise Harvesting:} In the first step, the device gathers a sequence of samples from an environmental noise process.
    This is usually done by a microcontroller with an analog-to-digital converter.
    These samples are typically filtered to remove undesirable features and time synchronized by sending messages over a public channel.
    \item \textbf{Bit Extraction (the focus of this paper):} Raw samples gathered from the environmental noise process are then converted by a fuzzy extractor into a sequence of bits that will be a key.
    Generally, a few bit errors (1-10\%) between authenticating devices are permitted in the extracted bit sequence.
    The most popular bit extraction technique is to divide the sequence of raw noise samples into bins of 10-20 samples each.
    \item \textbf{Key Reconciliation:} After bit extraction, each device will have a sequence of bits in memory that represents a key.
    Even though the devices are located nearby one another, differences in their measurements of the raw environmental noise will have caused spurious errors in the extracted bit sequences.
    Key reconciliation is the process by which two nearby devices exchange messages with one another over a public channel to resolve those bit differences.
    At the end of this step, if both devices are honest and located in the same vicinity, they will each hold identical authentication keys that can be used to encrypt data or validate their identities to one another.
\end{enumerate}

\paragraph{Fuzzy Extractors}
Context-based authentication relies on the authenticating devices observing nearly identical environmental noise signals to base their keys.
But even two nearby devices may read different values from their respective sensors during an event that creates noise.
If the event is closer to one device than the other, they will observe slightly different noise patterns and generate different keys.
Fuzzy extractors account for inconsistencies in observed environmental noise by examining noise and use a mapping function, which all the devices in the system share, to map the observed noise to a new value.

\paragraph{NIST Test for Randomness}
\label{sec:nist}

The NIST test for randomness is a software suite that evaluates the quality of a random bit sequence.
It consists of 15 separate tests that analyze the bit stream bit-wise, block-wise, and superblock-wise.
{\name}'s goal is to modify a bit sequence to increase the number of NIST tests that it passes.
Passing more NIST tests should be the goal of any random number generator to verify that generated keys are random.
We use the NIST tests for randomness as a benchmark for key quality.

The NIST test suite's input is a bit sequence, typically thousands to hundreds of thousands of bits in length.
For each test in the suite, the input bit sequence is divided into blocks, and each block is evaluated independently.
The output of a typical run of the suite, shown in Table \ref{tab:nist}, lists two crucial figures for each of the 15 tests: a $p$-value, and the proportion of bits, blocks, or superblocks that passed the test.

$p$ is the probability that a true random number generator would have produced a less random output than the given test input sequence.
$p$ values closer to $1$ are better, and $p$ values closer to $0$ are worse.
$C_1, C_2, ... C_{10}$ represent the number of $p$ values that lie in the intervals $[0.0,0.1),$ $[0.1,0.2),$ $...,$ $[0.9,1.0)$.

The \emph{proportion} output gives the fraction of blocks that passed each test in the suite.
We want the proportion of passes to be as close to unity as possible.
In general, the tests in the NIST suite need a minimum of 1000 bits to evaluate the quality of the sequence---shorter sequences cannot be evaluated with high confidence.
Each stream of bits needs to be at least 1000 bits, or our NIST test suite fails internally.
Therefore, if enough bits exist, we divide the stream into 100 blocks of at least 1000 bits in size.
The test suite divides sequences into blocks of 100 bits each and subjects each block to the suite of 15 tests.

\subsection{A Review of the Typical Set and the Asymptotic Equipartition Property}
\label{sec:aep}

\begin{figure}[t]
    \centering
    \includegraphics[width=\hsize]{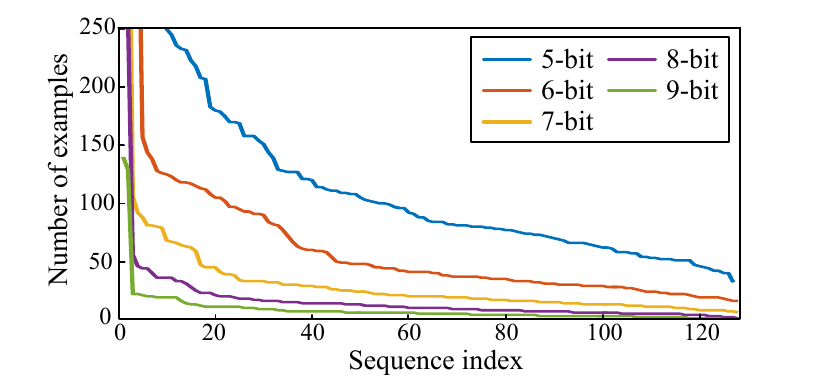}
    \vskip -10pt
    \caption{A sorted histogram of variable-length bit sequences generated from measured environmental noise.}
    \label{fig:sortedhist}
\end{figure}

One of the most significant characteristics of cryptographic keys is that their bit sequences be uniformly distributed---that is that each bit sequence should be equally probable.
If some bit sequences in the key are more likely than others, it would be easy for an attacker to guess the chosen key.
But sources that generate independent and uniformly distributed random numbers are typically challenging to build.
As previous work in context-based authentication has demonstrated, \textbf{most environmental noise does not yield uniformly distributed samples}.
How can we generate uniformly distributed numbers from a nonuniform source?

Suppose we independently sample a nonuniform source many times in succession.
The sequences of samples we obtain from that sampling process can be divided into a typical set and a non-typical set.
The probability that a sequence of independent samples lies in the typical set approaches 1 for sufficiently long sequences.
Furthermore, all sequences in the typical set are almost uniformly distributed.
This result is called the asymptotic equipartition property.
Figure \ref{fig:sortedhist} shows histograms of bit sequences of different lengths, sorted from most probable to least probable.
Bit sequences do not need to be very long before they begin to exhibit this almost-uniform characteristic.

\subsection{R\'{e}nyi Entropy}
\label{sec:renyi}

\begin{definition}{R\'{e}nyi Entropy}
Let $X$ be a random variable with an alphabet $\mathcal{X}$ and distribution $p_X(x)$.
The R\'{e}nyi Entropy of $X$ is defined as

\begin{equation*}
    R(X) = -log_2 \sum_{x \in \mathcal{X}} p_X (x) ^2
\end{equation*}

Bennett et. al. demonstrated~\cite{privacyamp} that the R\'{e}nyi entropy is a lower bound on the number of bits of private information that can be distilled after key reconciliation.
More generally, R\'{e}nyi entropy is a convenient tool for estimating the uncertainty in a random variable.
By Jensen's inequality, we have that  R\'{e}nyi entropy is upper bounded by the Shannon entropy:

\begin{equation*}
    R(X) \leq H(X)
\end{equation*}

\noindent
with equality when $X \sim Unif(\mathcal{X})$.

\end{definition}

\subsection{Related Work}\label{sec:related_works}

The literature that is adjacent to our work can be roughly separated into two categories: context-based authentication systems and cryptography.

\paragraph{Context-Based Authentication}
This body of work includes full-system implementations of context-based authentication systems.
In general, each system includes a mechanism for measuring environmental randomness, extracting bit sequences, and resolving bit errors in the bit sequences to form identical keys.
Many different sources of environmental randomness have been studied.
For body area networks of wearable devices, the H2H~\cite{h2h}, H2B~\cite{Lin-H2B}, and others~\cite{Zhang-ECG,BELKHOUJA2019109} systems measure ECG signals (heartbeat data).
Secret From Muscle~\cite{Yang-EMG} EMG (produced by skeletal muscles) and skin vibration has been used to generate keys between low-cost wearable devices and implantable medical devices.
Context-based authentication systems targeted at stationary IoT devices have used context from audio, humidity, luminosity, visual, and vibration channels~\cite{Miettinen-CCS14,Schurmann-TMC13,Saxena-SP06,Miettinen-DAC18, Shrestha-FCDS14, Lee-WiSec20}.
ProxiMate, Amigo, Ensemble, and others~\cite{proximate,Amigo,Ensemble,Hershey-unconventional,Jana-RSSI} extract entropy from measurements of the radio frequency spectrum.
Because the randomness density of environmental signals tends to be low, context-based authentication systems generally pass half or less than half of the NIST tests.
We have compiled the results of the NIST tests from a subset of the context-based authentication systems in Table \ref{tab:nist}\footnote{Not all context-based authentication systems publish their NIST test results. The results shown in Table \ref{tab:nist} are all the published results that we could find.}.

\paragraph{Cryptography}
In the domain of cryptography and information theory, a lot of work has focused on understanding the information content of signals.
Key reconciliation---a suite of techniques that context-based authentication relies heavily on---was originally developed for exchanging information over quantum communication channels~\cite{cachin97,privacyamp,publicdiscussion}.
Fuzzy extractors~\cite{fuzzy_extractor}, fuzzy vaults~\cite{Fuzzy-Vault}, fuzzy commitment~\cite{fuzzy_commitment}, and fPAKEs~\cite{fpake} are more modern cryptographic techniques that are commonly used in key reconciliation by context-based authentication systems.
Also other authors have built techniques to create a more uniform distribution of biased random number generators~\cite{drng}.

\begin{table*}
\caption{NIST test results for various context-based authentication schemes (\checkmark indicates pass).}
\center
\vskip -15pt
\begin{tabular}{lllllllllllllllll}
 & \rot{\textbf{Pass Rate}} &
\rot{\textbf{Frequency}}                 & \rot{\textbf{Block Frequency}} & \rot{\textbf{Cumul. Sums, Fwd}} &
\rot{\textbf{Cumul. Sums, Rev}}  & \rot{\textbf{Runs}}            & \rot{\textbf{Longest Run}}  &
\rot{\textbf{Rank}}                      & \rot{\textbf{FFT}}             & \rot{\textbf{Non-overlap Template}} &
\rot{\textbf{Overlap Template}}      & \rot{\textbf{Universal}}       & \rot{\textbf{Approx Entropy}} &
\rot{\textbf{Serial}}                    & \rot{\textbf{Serial}}          & \rot{\textbf{Linear Complex}} \\ 
\hline
Jana et al.~\cite{Jana-RSSI} & \multicolumn{1}{|r|}{10/15} & \checkmark & \checkmark  & \checkmark & \checkmark & \checkmark & \checkmark & & \checkmark & & & & \checkmark & \checkmark & \checkmark & \\ \hline
H2B~\cite{Lin-H2B}           & \multicolumn{1}{|r|}{8/15} & \checkmark & \checkmark  & \checkmark &  &  & \checkmark & & \checkmark & \checkmark & & & \checkmark & & & \checkmark \\ \hline
H2H~\cite{h2h}               & \multicolumn{1}{|r|}{8/15} & \checkmark & & & & \checkmark & \checkmark & \checkmark & \checkmark &  & & \checkmark & \checkmark & & & \checkmark \\ \hline
Xi et al. ~\cite{xi16}       & \multicolumn{1}{|r|}{10/15} & \checkmark & \checkmark & \checkmark  & \checkmark & \checkmark & \checkmark &  & \checkmark &  & &  & \checkmark & \checkmark & \checkmark & \\ \hline
Secret from Muscle ~\cite{Yang-EMG} & \multicolumn{1}{|r|}{9/15} &\checkmark & \checkmark & \checkmark  & \checkmark & \checkmark & \checkmark &  & \checkmark &  & &  & \checkmark & \checkmark & & \\ \hline
VoltKey~\cite{voltkey}       & \multicolumn{1}{|r|}{8/15}&\checkmark & \checkmark  & \checkmark & \checkmark  &  &  & \checkmark &  & \checkmark & & & \checkmark & & & \checkmark \\ \hline \hline
\textsc{VoltKey} + Von Neumann corr. & \multicolumn{1}{|r|}{11/15}&\checkmark & \checkmark  & \checkmark & \checkmark  &  &  & \checkmark & \checkmark & \checkmark &  & & \checkmark &\checkmark & \checkmark & \checkmark \\ \hline
\textsc{VoltKey} + {\name}             & \multicolumn{1}{|r|}{14/15}&\checkmark & \checkmark & \checkmark & \checkmark  & \checkmark & \checkmark & \checkmark & \checkmark & \checkmark & \checkmark &  & \checkmark & \checkmark & \checkmark & \checkmark \\ \hline
\end{tabular}
\vskip -10pt
\label{tab:nist}
\end{table*}

\section{Entropy Rate of a Noisy Bandlimited Signal}

This section develops a new technique for measuring the entropy rate of a bandlimited signal from its power spectral density (PSD).
We begin with the Burg Max Entropy Theorem, which compares a signal's PSD and its entropy rate.
But as we will see, Burg's theorem is intractable to compute for signals more than a few samples in length.
We develop a computationally-tractable method for calculating the entropy rate of reasonably-sized signals (several thousand samples long).

We could use the Shannon-Hartley Channel Capacity theorem, which relates the noisy signal's entropy rate to an SNR.
The difficulty with this method is that it is not clear how to interpret the SNR in situations where we are dealing only with noise.
In most context-based authentication scenarios, we want as much uncorrelated noise as possible, and we try to filter out everything else.
We would prefer to have an expression for entropy rate that is a function of the power spectral density of the harvested randomness signal because that does not require us to measure the SNR.

\subsection{System Model}
Suppose we harvest entropy from some environmental signal that is composed of a deterministic part and a random part:

\begin{equation}
    X(t) = D(t) + Z(t)
\end{equation}

Where $D(t)$ is an unknown bandlimited signal deterministic in time and $Z(t)$ is additive white Gaussian noise (AWGN).
$X(t)$, $D(t)$ and $Z(t)$ are assumed to be continuous in time.
This model is used by many context-based authentication schemes, including \textsc{VoltKey}~\cite{voltkey}, H2H and H2B~\cite{h2h,Lin-H2B}.
We can represent $D(t)$ as an expansion in the basis of sinusoids:

\begin{equation*}
    D(t) = \sum_{k = 0}^{N} a_k cos(2 \pi k f t + \Theta)
\end{equation*}

\noindent
Where the $a_k$s are expansion coefficients and $\Theta_k$ is the phase angle, modelled as a uniformly distributed random variable on $[- \pi,\pi]$.

\begin{theorem}
\label{thm:stationarity}
$D(t)$ is a stationary process.

\end{theorem}
\begin{proof}
See Appendix \ref{thm:stationarityproof}
\end{proof}

The autocorrelation function is the same for all time shifts, meaning that its value is only dependent on the lag $t_1 - t_2$, not on the absolute time $t_1$ or $t_2$.
We have verified that the ACF depends only on time shift for functions of the same form with (a) more than two terms and (b) various combinations of coefficients on each term.
This is the criterion for stationarity.

We model the entropy source $X(t)$ as a stochastic process that can be sampled discretely in time, yielding a sequence $\{X_m\}$.
In this paper, \textbf{$\{X_m\}$ is the sequence of ADC samples of the environmental noise process}.

\subsection{Computing the Entropy Rate from the PSD}

Our goal is to get a bound on the source process's entropy rate $\{X_i\}$, which can be directly computed from the process's power spectral density.
We want an inequality similar to the Shannon-Hartley channel capacity bound but where the signal bandwidth and SNR are directly computed from the properties of the source process's PSD.

Assuming that the entropy source is a wide-sense stationary (WSS) stochastic process (as we demonstrated in Theorem \ref{thm:stationarity}), we can compute its autocorrelation function by taking the inverse Fourier transform of the PSD~\cite{gubner}: $R_{XX}(\tau) = \mathcal{F}^{-1}(S(f))$\footnote{This is called the Wiener–Khinchin theorem.}:

\begin{theorem}{Burg's Maximum Entropy Theorem}

In general, the maximum entropy rate stochastic process $\{X_i\}$ satisfying the constraints

\begin{equation}
\label{eqn:burgconditions}
    R_{XX}(k) = E \left[ X_m X_{m+k} \right] = \alpha_k, \text{~~} k = 0, 1, 2, ..., p
\end{equation}
 is the $p$th-order Gauss-Markov process of the form
 
\begin{equation}
    X_m = -\sum_{k = 1}^p a_k X_{m-k} + Z_m
\end{equation}
\end{theorem}

\begin{proof}
See \cite{coverthomas}.
\end{proof}
In other words, each new sample $X_m$ is a linear combination of the previous $p$ samples plus some iid additive white Gaussian noise $Z_m \sim \mathcal{N}(0, \sigma ^2)$.
The $a_k$s are chosen to satisfy Equation \ref{eqn:burgconditions}.
We can use the Yule-Walker equations to calculate the $a_k$s from $\sigma$ and the $\alpha _k$s, the values of the autocorrelation function of the source process $\{X_i\}$~\cite{coverthomas}.

Burg's maximum entropy theorem holds \emph{without} assuming that $\{X_m\}$ is broad sense stationary.
However, to compute the $\alpha _k$s from the PSD, we need the source process $D(t)$ to be wide sense stationary.
We can write each $X_{m-k}$ as an expansion in the basis of complex sinusoids:

\begin{equation}
\label{eqn:maxentropy}
    X_m = -\sum_{k = 1}^p \sum_{l=0}^N a_k b_l e^{-i 2 \pi l f m/N} + Z_m
\end{equation}

In Theorem \ref{thm:stationarity}, we demonstrated that the source process of interest is stationary under some reasonable assumptions, and therefore it satisfies the conditions in Equation \ref{eqn:burgconditions}.
In fact, the $\alpha _k$s from Equation \ref{eqn:burgconditions} can be directly read from the autocorrelation function of the source process, which is the inverse Fourier transform of the PSD.

\subsection{Computing R\'{e}nyi Entropy Rate}
\label{sec:entropyrate}

Plugging in our model for $X_m$ from Equation \ref{eqn:maxentropy} into the expression for R\'{e}nyi Entropy:

\begin{align*}
    R(X_m) &= -\log_2 E \left[ p_X(x)^2 \right] \\
           &= -\log_2 \sum_{x \in \mathcal{X}} Pr \left[ -\left( \sum_{k=1}^p \sum_{l=0}^N a_k b_l e^\frac{-i 2 \pi l f m}{N} + Z_m \right) = x\right]^2
\end{align*}

Where in the second step, we plug in the discrete Fourier transform representation of our signal.
We will now modify the limits on the double sum.
First, we will set $p=N$, so that the Markov process's order is a complete sampling window.
Second, we will drop the $l = 0$ term in the DFT, assuming no DC offset, that the sample process has zero mean.
This allows us to combine the double sum into one sum that runs from 1 to $N$.

\begin{align*}
    R(X_m) &= -\log_2 \sum_{x  = 0}^{2^b-1} Pr \left[ -\left( \sum_{k=1}^N a_k b_k e^\frac{-i 2 \pi l f m}{N} + Z_m \right) = x\right]^2
\end{align*}

Here, the $a_k = \sqrt{\alpha_k}$, the samples of the PSD of the deterministic signal.
The outer sum over $x$ runs from 0 to $2^b-1$ because we are assuming that the signal is being acquired with a $b$-bit analog to digital converter.

\subsubsection{If Deterministic Signal is 0}

Suppose the deterministic component of the signal $D(t)$ is zero, and the only source of entropy is $Z_m$, the additive white Gaussian noise, which we assume is correlated among two or more context-based authenticators.
For example, this would be the case in the H2H body area network when there is no ECG signal on the skin, or in VoltKey when there is no 120VAC power waveform.
The only noise present is caused by cosmic radiation.

This is a slightly unrealistic assumption, but it allows us to comprehend the amount of entropy carried by the correlated random noise, whose statistical properties we know.
In this calculation, we want to find the amount of entropy carried by the AWGN, not the amount of shared entropy---called mutual information---common to both devices.

\begin{theorem}{}
\label{thm:renyibound}
The R\'{e}nyi entropy of the samples of the source process $\{X_i\}$ is lower bounded by the R\'{e}nyi entropy of the AWGN $Z_m$. In other words, $R(X) \geq R(Z)$.
After observing the deterministic component $D$ of the source process, the remaining uncertainty in $X$ is due to $Z$, the AWGN.

\end{theorem}
\begin{proof}
See Appendix \ref{sec:renyisumproof}.
\end{proof}

\begin{theorem}
\label{thm:entropyrate}
Let $Z_m$ be a single sample of the analog source process $\{X_i\}$ consisting only of additive white Gaussian noise.
Then its  R\'{e}nyi entropy is $R(Z_m) \approx \log_2 2 \sigma \sqrt{\pi}$.

\end{theorem}

\begin{proof}

\begin{align*}
    R(X_m) &= -\log_2 \sum_{k=-2^{b-1}}^{2^{b-1}} Pr \left[ b_k \leq Z_m \leq b_{k+1} \right]^2 \\
           & \approx -\log_2 \sum_{k=-2^{b-1}}^{2^{b-1}}   \frac{1}{2 \pi \sigma ^2} e^{-k^2/\sigma ^2}  \\
           & \approx - \log_2 \int_{- \infty}^{+ \infty} \frac{1}{2 \pi \sigma ^2} e^{-k^2/\sigma ^2}   \mathrm{d}k\\
           & \approx \log_2 2 \sigma \sqrt{\pi}
\end{align*}
\end{proof}

Here we assume that the analog AWGN signal will be quantized by an analog-to-digital converter with $b$ bits of precision.
The sum over the alphabet $\mathcal{X}$ is over all of the $2^b$ possible quantizations that can be produced by a $b$-bit ADC.
In the last approximation, we are assuming that $\sigma$ is small, and we are capturing most of the probability of $Z_m$ within the limits of the ADC's dynamic range.

In a typical application where we are intentionally amplifying correlated random noise $Z_m$, we can expect $\sigma$ to be on the order of $1/10$th the dynamic range of our ADC.
For reasonable values of $\sigma$, the R\'{e}nyi entropy is limited to about 10--12 bits per sample.

This is a lower bound on the Shannon entropy of $X(t)$ for two reasons.
First, as discussed in \S \ref{sec:renyi}, $R(X) \leq H(X)$ for all random variables $X$.
Second, our computations in this section assume that the deterministic component of $X(t)$ is zero.
Adding a nonzero deterministic component will increase the entropy per sample (for proof of this claim, see Appendix \ref{sec:renyisumproof}).

Still, the R\'{e}nyi entropy rate of $Z_m$ is relatively high, even if we ignore the deterministic component of the signal.
Standard key extraction techniques used by context-based authentication mechanisms are only able to extract one bit of entropy per $\sim10$ ADC samples: a bit extraction rate about two orders of magnitude lower than what we would expect to be carried by AWGN noise process alone!

\subsubsection{If Deterministic Signal is Nonzero}
\label{sec:detnonzero}

In the case where $p$th order Gauss-Markov process, it is possible to compute the entropy rate without the Yule-Walker equations.

\begin{align}
    H(\mathcal{X}) &= H(X_p | X_{p-1}, ..., X_0) \\
    &= H(X_0, ..., X_{p}) - H(X_0, ..., X_{p-1}) \\
    &= \frac{1}{2}\log_2(2 \pi e)^{p+1} |K_p| -  \frac{1}{2}\log_2(2 \pi e)^{p} |K_{p-1}| \\
    &=  \frac{1}{2}\log_2 \left((2 \pi e) \frac{|K_p|}{|K_{p-1}|}\right)
    \label{eqn:entropyrate}
\end{align}

where $K_p$ is the autocorrelation matrix of the process $\{X_m\}$.
\noindent
The $K_p$ autocorrelation matrix is called a Toeplitz matrix.
It is rank $N$.
Its top row is the individual values of the autocorrelation function of the source process $\{X_k\}$.
Using the fact that

\begin{equation*}
    E[X_k X_l] = R_{XX}(k-l) = R_{XX}(l-k)
\end{equation*}
for wide-sense stationary processes:

\[
=
\begin{bmatrix}
    R_{XX}(0) & R_{XX}(1) & R_{XX}(2) & \dots & R_{XX}(N-1) \\
    R_{XX}(1) & R_{XX}(0) & R_{XX}(1) & \dots & R_{XX}(N-1) \\
    R_{XX}(2) & R_{XX}(1) & R_{XX}(0) & \dots & R_{XX}(N-2) \\
    \hdotsfor{5} \\
    R_{XX}(N-1) & R_{XX}(N-2) & R_{XX}(N-3) & \dots & R_{XX}(0)
\end{bmatrix}
\]

Equation \ref{eqn:entropyrate} gives the entropy rate in terms of the source process's autocorrelation function.
Unfortunately, Equation \ref{eqn:entropyrate} presents a fairly serious problem: the above result cannot be computed directly for large values of N because the determinant of the autocorrelation matrix approaches zero as N increases.
Computing the determinant of large matrices is a problem in general and not a pathology that is isolated to our autocorrelation matrix.
We discuss below a method for computing the ratio $|K_p|/|K_{p-1}|$ without directly computing the determinants.

\subsection{Calculating Entropy Rate without Directly Computing Determinants}
\label{sec:fancyentropyrate}

The problem that we encounter when attempting to compute an entropy rate from Equation \ref{eqn:entropyrate} is that the columns of the autocorrelation matrix $K_p$ are almost linearly dependent.
This makes the determinants of $K_p$ and $K_{p-1}$ close to but not exactly zero, and the computer's floating point representation rounds the result of the determinant computation to zero.
Our goal here is to compute the \emph{ratio} of the determinants, so we are not concerned about the fact that they are individaully small.

\begin{lemma}
\label{lemma:uppertriangulardeterminant}
If $R$ is an $N \times N$ triangular matrix, then the determinant of $R$ is the product of its diagonal elements:
\begin{equation*}
    |R| = \prod_{i=1}^N r_{ii}
\end{equation*}
\end{lemma}

$|K_p|$ and $|R|$ are close to zero because all of its diagonal elements are less than 1, causing the product in Lemma \ref{lemma:uppertriangulardeterminant} to approach zero for large $N$.
To be clear, $R$ and $K_p$  are both full-rank matrices.

\begin{lemma}
If $\mathbf{A} = \mathbf{Q}_1\mathbf{R}_1 = \mathbf{Q}_2\mathbf{R}_2$ are two QR decompositions of full rank square matrix $\mathbf{A}$, then

\begin{equation*}
\begin{split}
\mathbf{Q}_2 = \mathbf{Q}_1 \mathbf{S}
\\
\mathbf{R}_2 = \mathbf{S} \mathbf{R}_1
\end{split}
\end{equation*}

\noindent
for some square diagonal $\mathbf{S}$ with entries $\pm 1$.
If we require the diagonal elements of $\mathbf{R}$ to be positive, then the factorization is unique.
\end{lemma}
\begin{proof}

Starting with the factorization $\mathbf{A} = \mathbf{Q}_1\mathbf{R}_1 = \mathbf{Q}_2\mathbf{R}_2$ we can define a new matrix $\mathbf{S}$ in the following way:

\begin{equation*}
\mathbf{S} = \mathbf{Q}_2^T \mathbf{Q}_1 = \mathbf{R}_2 \mathbf{R}_1^{-1}
\end{equation*}

Since $\mathbf{Q}_1$ and $\mathbf{Q}_2$ are unitary, then $\mathbf{S}$ must also be unitary.
Since $\mathbf{R}_1$ and $\mathbf{R}_2^{-1}$ are upper triangular, then $\mathbf{S}$ must also be upper triangular.
This means that $\mathbf{S}$ is a diagonal matrix with elements $\pm 1$.

\begin{equation*}
\mathbf{Q}_1\mathbf{R}_1 = \mathbf{Q}_2\mathbf{R}_2
\end{equation*}

\begin{equation*}
\mathbf{Q}_1\mathbf{R}_1\mathbf{R}_2^{-1} = \mathbf{Q}_1\mathbf{S} = \mathbf{Q}_2
\end{equation*}

Now, apply the constraint that $\mathbf{R}_1$ and $\mathbf{R}_2$ have positive entries on their diagonals, which forces $\mathbf{S} = \mathbf{I}$ and $\mathbf{R}_1 = \mathbf{R}_2$.

\end{proof}

\begin{theorem}

The ratio $|\mathbf{K}_p| / |\mathbf{K}_{p-1}| = r_{pp}$, the lower rightmost element of $\mathbf{R}$ in the QR factorization of $\mathbf{K}_p$.
\end{theorem}
\begin{proof}

\begin{equation*}
\prod_{k = 1}^p r_{kk} = 10^{\log \prod_{k = 1}^p r_{kk}} = 10^{\sum_{k = 1}^p \log r_{kk}}
\end{equation*}

\begin{equation*}
\frac{|K_p|}{|K_{p-1}|} = \frac{10^{\sum_{k = 1}^p \log r_{kk}}}{10^{\sum_{k = 1}^{p-1} \log r_{kk}}} = 10^{\sum_{k = 1}^p \log r_{kk} - \sum_{k = 1}^{p-1} \log r_{kk}} = r_{pp}
\end{equation*}

\end{proof}

Note that there is a unique QR factorization of $K_p$ that has all positive diagonal elements of $R$.
We need $R$ to have positive elements on its diagonal in order to be able to take their logarithm.
Using the above technique for computing the ratio $|K_p|/|K_{p-1}|$, we can rewrite Equation \ref{eqn:entropyrate} by bringing the factor of $1/2$ into the logarithm and assuming that the source process consists only of uncorrelated Gaussian noise:

\begin{equation*}
H(\mathcal{X}) = \log_2 \sigma \sqrt{2 \pi e}
\end{equation*}

Where we replace the ratio of determinants with the standard deviation of the source process.
The Shannon entropy differs from the R\'{e}nyi entropy only by a factor of $\sqrt{e/2}$.

\subsection{Measuring Shared Randomness}

The idea that underlies context-based authentication techniques is that two devices can build a shared key from a shared randomness source, which is not observable by third parties.
Evaluating the key generation scheme matters not just how much randomness is encoded in the signal to begin with but how much of that randomness is common to both devices.

The challenge in doing a theoretical analysis of the amount of mutual information that we would expect to be common to two context-based authenticators is that it tends to be situationally dependent.
Different environments, different physical configurations of devices, and other parameters weigh heavily on the correlation between two source processes.

To compare the amount of mutual information between two environmental signals to the amount of entropy, we collected some voltage measurements at high frequency from the power outlets in our institution's offices using \textsc{VoltKey} prototypes and analyzed their characteristics.
We used standard techniques from \S \ref{sec:contextoverview} to sample and time-align voltage measurements without key extraction or reconciliation.
We split each signal into blocks of 150,000 samples (about two seconds of data) and computed each block's mutual information and entropy.

There is an order of magnitude difference in entropy and mutual information---the amount of entropy common to two authenticating devices---for the signal that we tested in this experiment.
Said another way, only about 10\% of the randomness measured by the \textsc{VoltKey}s is common to both devices.
The entropy that is not common to both devices cannot be used to generate a key.

The entropy rate of the source process is about 16 bits per sample---slightly higher than we would expect if we were gathering uncorrelated Gaussian noise.
In Theorem \ref{thm:entropyrate}, we concluded that we would expect roughly 10--11 bits of entropy per sample.

Based on our experience working with context-based authentication methods, we think that this order-of-magnitude-gap between entropy and mutual information is probably typical of many noise types.
However, it is not easy to be confident without conducting a more formal analysis.

What is surprising is that conventional algorithms for extracting keys from a natural environmental noise process only generate bits at about $1/10$th the rate of the mutual information in this experiment.
We should expect the bit extraction rates to be slightly lower than the mutual information rate to avoid errors---perhaps $1/2$ or $1/3$ bit per sample---but this discrepancy is inefficient.
Even with such low bit extraction rates, most context-based authentication schemes still have to perform key reconciliation to eliminate bit errors in the shared key.

As we will see, the raw environmental signal's entropy rate affects the performance of the randomness distiller.

\section{Randomness Distillation}

Once we have extracted a bit sequence from an environmental source noise process, we need to distill its entropy to pass standard tests for randomness.
Raw bit sequences extracted from the environment often do not pass standard randomness tests (Table \ref{tab:nist}).

{\name} is a randomness corrector that transforms independent identically distributed samples from a random source to make their distribution closer to uniform.
It works by concatenating samples of the source into sequences.
By the asymptotic equipartition property, \emph{sequences} of samples will be nearly uniformly distributed even if the individual samples are not.
Von Neumann's corrector is a classical technique that aims to accomplish the same goal.
We compared the preformance of {\name} to Von Neumann's algorithm.

\paragraph{Von Neumann Corrector}

The Von Neumann Corrector was an early technique used to normalize the histogram of randomly generated bit sequences~\cite{vonneumann}.
Its goal is to generate bit sequences in which 1 and 0 are equally likely to occur from an input sequence of unfair coin tosses.
The Von Neumann Corrector groups the input sequence into pairs of bits, and it discards pairs in which both bits are the same ($\{1,1\}$ and $\{0,0\}$).
For pairs of bits that are not the same, The Von Neumann Corrector copies only the first bit in the sequence to the output.
It is essentially an application of a special case of the Asymptotic Equipartition Property, for which the sequence length is 2.
The Von Neumann corrector did not improve the NIST test pass rate when we applied it to the raw \textsc{VoltKey} bit sequences (see Table \ref{tab:nist}) enough to use its output as a cryptographic key.
We need a new technique to improve the quality of keys generated by context-based authentication systems.

{\name} is a new entropy distillation technique that converts long bit sequences with low entropy per bit into shorter bit sequences with high entropy.
{\name} takes as input a bit sequence from a random generator and groups bits into blocks $k$.
Each block of $k$ bits represents an integer in the range $[0,2^k-1]$.
{\name} then generates a histogram of all $k$-bit integers obtained from the raw input sequence.
The resulting histogram is divided into two categories: the typical set, which is almost uniformly distributed, and the non-typical set, which occurs either much more frequently or much less frequently than average.
Half of the input sequences are assumed to be part of the typical set, which is retained, and the other half of the sequences are discarded.

After discarding the non-typical set, we can represent the remaining elements of the typical set using $k-1$ bit sequences.
The new $k-1$ bit indices are the output of the corrector.
Fig. \ref{fig:histograms}(b) shows the histogram of the raw bits extracted from a \textsc{VoltKey} device grouped into 7-bit blocks and assigned indices in the range $[0,127]$.
Fig. \ref{fig:histograms}(a) shows the histogram of data after it has been corrected by {\name}.
For comparison, Fig. \ref{fig:histograms}(c) shows a histogram of 7-bit numbers generated by python's random number generator.

In Fig. \ref{fig:bit_extraction}(a), we plot the number of bits remaining after correction by {\name} as a function of the sequence length.
The Von Neumann corrector discards over 80\% of the bits in the input sequence.
By tuning the subsequence length, {\name}'s bit usage can be adapted into any size of data given by the user.

\subsection{Our Implementation of {\name}}
{\name} examines a binary stream, partitions that stream into bit sequences, and creates a histogram of the partitioned bit sequences.
{\name} separates the extracted sequences into two categories: typical sequences and non-typical sequences.
It then discards non-typical sequences and maps the typical sequences to new values based on those new histograms.
The key components of our algorithm are:
\begin{enumerate}
    \item Sequences of size $k$ get remapped into sequences of size $k-1$.
    \item We drop bits after each individual sequence from the bit stream. We discuss the reasons for doing this at the end of this section.
    \item Bit sequences that occur most frequently on the input stream are dropped. The remaining sequences (the typical set) are remapped to new bit sequences.
\end{enumerate}

\noindent
The details of our algorithm are below.
\begin{enumerate}[leftmargin=0cm,itemindent=.5cm,labelwidth=\itemindent,labelsep=0cm,align=left]
\item \textbf{Partition Bit Stream into Sequences} {\name} is given a bit stream, while a system is running, generated from environmental noise as input.
We first want to partition the bit stream into subsequences of length $k$.
If we let $k=8$, we then initialize four arrays two of them are of size $2^k$ while the other two arrays are size $2^{k-1}$.
These arrays each have different jobs.
One of the arrays (called \texttt{A}) keeps track of the number of occurrences of each subsequence in the input bit stream.
\texttt{A} holds a histogram of integers extracted from the environmental noise source, similar to that shown in Fig. \ref{fig:histograms}.
As we read bits from the input stream, we use array \texttt{B} to keep track of the order in which particular subsequences first occur.
After converting a subsequence from the input of length $k$ to an integer and updating arrays \texttt{A} and \texttt{B}, we skip $m$ bits in the input sequence.

\item \textbf{Find Highest Occurring Half of Sequences} To extract the typical set, we throw away the most commonly occurring subsequences of bits.
To do that we compare the number of occurrences of each subsequence to all others in \texttt{A}.
After knowing which subseqeunces are most common, we drop the most common.
If a index is found to be in the highest half we change that index in \texttt{A} to be $-1$.
We mark the same indices in \texttt{B}.

\item \textbf{Remap Binary Sequences and Return the Changed Stream} The last array, \texttt{C}, will be used to hold the remappings.
Now that we know which values to keep, we iterate through \texttt{B} and if the value is not marked as $-1$ remap it to the next available integer in the \texttt{C}.
If a number appears in the stream a second time that hasn't been remapped too, we throw it out.
After we preform these steps the result will be a bit extracted stream.
We also skip over the same amount of bits we did in the first step.
For example, if 192 is the first number that appears in the bit stream, and when indexed in array \texttt{A} the value is not $-1$.
Whenever we see the value 192 in the bit stream, we replace that value with its remapped equivalent that is defined in \texttt{C}.
So, every occurrence of 192 in the input data stream will be replaced with zero.
Also, we are replacing an original bit stream of size $k$ with a new bit stream of size $k-1$.
Therefore dropping the last bit.
\end{enumerate}

\paragraph{Real Environmental Signals}
\label{sec:realsignals}

Our analysis in \S \ref{sec:fancyentropyrate} assumes that the environmental signal is a stationary process, which is a valid assumption during periods of everyday activity.
But during periods of little or no activity, many environmental noise signals are not stationary.
During inactivity, the entropy rate is effectively zero, causing samples of the source noise process not to be independent or identically distributed.
The result is that we get very long runs of zeros or ones in our extracted bit sequences during periods of inactivity.
We can repair the bit sequence by dropping bits from the input.

\emph{Bit dropping} has the effect of shortening long runs of ones or zeros during periods of inactivity.
It can not make a non-stationary process stationary, but it can reduce the length of non-stationary sequences, making them less influential in the overall signal's properties.

An alternative we considered is to wait until the source noise process has a high enough entropy rate before beginning key generation.
But waiting for the environmental entropy rate to increase above a threshold is impractical for real systems because two authenticating devices may not agree on the exact moment when the entropy rate becomes high enough.
Coordinating between multiple devices would require additional communication, which wastes time.
Another advantage of bit dropping is that it allows keys to be generated immediately when requested by the user rather than waiting for an acceptable entropy rate in the environmental noise.

\subsection{Choosing Parameters $k$ and $m$}

The quality of keys generated by \textsc{Moonshine} depends heavily on our parameters in Algorithm \ref{algo:max}.
The relationship between parameter values and key quality is generally monotonic---higher parameter values usually produce better keys.
Larger values of $k$ and $m$ cause \textsc{Moonshine} to analyze the input datastream in longer blocks, reducing the similarity of nearby patterns.
The type of environmental source (audio, voltage, etc) also can change the relationship between parameter value and key quality.

In our implementation of \textsc{Moonshine}, an authenticator analyzes the noise source in real time to find the parameter values that maximize key quality and shares those parameters with the other device that is trying to pair.
Both devices then use the shared parameter values to apply Algorithm \ref{algo:max}, generating a shared key.

{\name} uses a warmup period during which the device continuously samples the environmental noise source, building a bit sequence histogram.
Once one device has collected enough data to characterize the properties of the environmental noise---encoded by the histogram---key generation commences.
We use the histogram generated during warmup to define mappings from input bit sequences to output bit sequences during key generation.

\newcommand\mycommfont[1]{\footnotesize\ttfamily\textcolor{gray}{#1}}
\SetCommentSty{mycommfont}
\begin{algorithm}
\DontPrintSemicolon
\KwIn{A bit sequence $B=\{b_1, b_2, \ldots, b_N\}$}
\KwIn{Number of bits in a subsequence $k$}
\KwIn{Number of bits to skip between subsequences $m$}
\KwOut{A mapping from $B \rightarrow \mathbb{Z}_{N/2}$}
\For{$j \gets 1$ \textbf{to} $k$} {
  $sequences[j].frequency \gets 0$\tcp*{Initialization}
  $sequences[j].value \gets 0$\;
  $sequences[j].order \gets NULL$\;
}
\tcc*{Count the frequency of occurrences of subsequences.}
\For{$i \gets 1$ \textbf{to} $N$ \textbf{by} $m+k$} {
  $subseq \gets \{b_i, b_{i+1}, \ldots, b_{i+k-1}\}$\tcp*{subseq is k-bit int}
  $sequences[subseq].frequency++$\;
  $sequences[subseq].value \gets subseq$\;
  \If{$sequences[subseq].order \neq NULL$} {
    $sequences[subseq].order \gets i/k$;
  }
}
$sequences \gets sort(sequences)$\tcp*{sort by frequency ascending}
$sequences \gets sequences[1..N/2]$\tcp*{Toss most-freq sequences}
\Return{$sequences$}\;
\caption{Algorithm for {\name}.}
\label{algo:max}
\end{algorithm}

\section{Evaluation of {\name}}
\label{sec:eval}

In this section, we evaluate the performance of {\name} with input data gathered from various forms of context-based authentication in real-world environments. Code is available from our GitHub repository\footnote{\url{https://github.com/jweezy24/Moonshine}}.
Our ZIA datasets are available from~\cite{MoonshineZenodo}.

We show that {\name} can generate high-quality keys from environmental noise sources with relatively low entropy rates, making them robust against attacks.
Before presenting evaluation results, the following details the hardware and parameter settings as well as main evaluation metrics.

\paragraph{Datasets}
The authors of \textsc{VoltKey}~\cite{voltkey} lent us prototypes of their hardware, which allowed us to benchmark our prototype implementation of {\name} on a realistic hardware platform.
In addition to data gathered by \textsc{VoltKey} hardware, we used the ZIA datasets published by Fomichev et al. ~\cite{ZIADatasets}.
The dataset consists of data streams collected from seven types of sensors---acceleration, luminosity, temperature, humidity, barometric pressure, magnetometer, gyroscope within different contexts.
It includes long sequences of synchronized data from each sensor in an office setting, a mobile device, and a car.
The authors preprocessed the data to generate bit sequences for each setting (Office, Mobile, Car).
All preprocessed bit sequences from the dataset have low entropy density, caused mainly by long periods of inactivity in the underlying source noise processes.
None of the bit sequences in the Office, Mobile, or Car datasets pass the NIST tests before being processed by {\name}.

We also generated a dataset of radio frequency measurements similar to those published in \cite{Jana-RSSI,proximate}.
The hardware we used to collect the RF dataset consisted of a short stub of wire connected directly to the analog-to-digital converter input of a microcontroller.
Our RF measurements produced relative high-quality bit sequences which passed most of the NIST tests before being processed by {\name}.

The audio dataset~\cite{Audio_dataset} as well as the mobile, office, and car datasets~\cite{ZIA_datasets} came from publicly available postings.
To evaluate the audio, we created the bit stream using a commonly used algorithm\cite{Audio_Alg}. 
Larger datasets are evaluated on a server-class machine because of memory restrictions on our Cortex M4. 
The server is also needed for NIST test evaluation, as the tests are not compatible with the Cortex M4 board.
We also evaluate our algorithm on \textsc{VoltKey} hardware in real-time to characterize its performance on an IoT-class platform.

\paragraph{Evaluation metrics}
We evaluate the performance of {\name} with various metrics.

\begin{itemize}
    \item NIST test pass rate: the fraction of NIST tests (discussed in \S \ref{sec:nist}) that pass. To generate high-quality keys, we want the NIST test pass rate to be as high as possible.
    \item Data retention rate: the fraction of input data retained in the output. We want to retain as much information as possible after processing by {\name}.
    \item Diversity of datasets: We used a wide range of data from varying sources. 
    
\end{itemize}

\subsection{Subsequence Length}

In this section, we want to understand how the key quality depends on the length of subsequences processed by {\name}.
According to the AEP (\S \ref{sec:aep}), longer subsequences should generate better keys.
In Fig. \ref{fig:bit_extraction}(b) we plot the NIST test pass rate as a function of subsequence length for {\name}-corrected data.
We can achieve a near-perfect pass rate for subsequences of at least 6 bits.
The same figure also plots the NIST test pass rate of the Von Neumann corrector for the same input data.

The tradeoff in {\name} is that the corrector discards a portion of the input data stream in the process of distilling its randomness.
Fig. \ref{fig:bit_extraction}(a) shows the amount of data remaining after randomness distillation.

Analyzing bits in longer subsequences increases the keys' quality because the typical set becomes more uniformly distributed on longer bit sequences.
As the sequences that {\name} considers become longer, there is a starker separation between the typical set and the non-typical set\footnote{This is the trend that is visualized in Fig. \ref{fig:sortedhist}.}.
For long sequences (8 or more bits), the non-typical set represents a smaller fraction of the overall data, and more of the sequences that we actually observe are elements of the typical set.
Since {\name} retains elements of the typical set and discards elements of the non-typical set, it makes sense that we would retain more overall data if we consider longer bit sequences.

\subsection{Run Time}

In Fig. \ref{fig:bit_extraction}(c), we plot the run time of {\name} on the \textsc{VoltKey} hardware, which uses an ARM Cortex M4 microcontroller running at 48 MHz.
{\name} does not add a substantial amount of time to the key generation process, even when processing long bit sequences.
By comparison, the sampling, bit extraction, and key reconciliation steps take about 20 seconds on a pair of \textsc{VoltKey}s.
This test ran on an input data stream of 200,000 bits gathered directly from the \textsc{VoltKey}.
The distilled bit sequence is in the tens of thousands of bits in length, much longer than generally used for an authentication key.
{\name} can distill randomness from the input data sequence in a reasonable amount of time.
The run time should not be substantially different when using different datasets as inputs.

\begin{figure}
    \centering
    \includegraphics[width=\hsize]{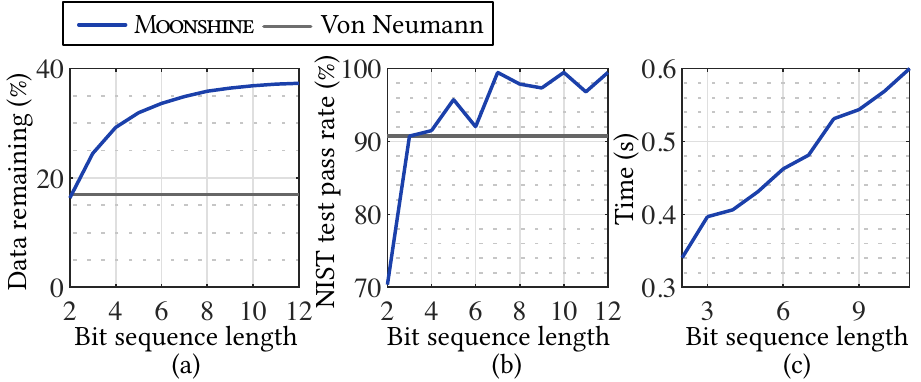}
    \vskip -10pt
    \caption{Evaluation of {\name} on \textsc{VoltKey}: (a) Percent of bits kept after {\name}. (b) Percentage of NIST tests that pass. (c) Runtime of {\name} on ARM Cortex-M4 microcontroller.}
    \label{fig:histograms}
    \vskip -10pt
\end{figure}

\subsection{Key Quality}

In this section, we evaluate the quality of keys generated by {\name} using the NIST test for randomness~\cite{NIST}.
{\name} makes two parameters available to the user: bit sequence length and bit drop length ($k$ and $m$ in Algorithm \ref{algo:max}).
Later evaluation will focus on the effect those parameters have on NIST test pass rates.
Fig. \ref{fig:discardlength} shows the NIST test pass rate as a function of $k$ and $m$ for the office, mobile, and car datasets, which generate raw data with similar entropy densities before being processed by {\name}. Overall, {\name} significantly improves the quality of keys generated from a wide variety of environmental noise sources.

\paragraph{Choosing Parameters $m$ and $k$}
The general trend is that larger parameter values tend to generate higher-quality keys, so we would advise users of {\name}  to choose larger parameter values to improve key quality.
This is because larger parameter values cause {\name} to process the input data stream in longer blocks, resulting in less local similarity.

An exception to the trend is in the Car and Mobile 1 datasets, in which the most significant parameter values cause almost all NIST tests to fail.
In those two datasets, the raw input datastream has such low randomness density already that {\name}  removes a substantial portion of the input bits when it is run with large parameter values.
This creates problems because the NIST test suite needs a minimum number of bits to determine whether a bit sequence passes each test.
The output data stream fails the NIST tests when {\name} removes too many bits.
We expect that keys generated from the Car and Mobile 1 datasets that fail with large parameter values would pass nearly all NIST tests if the datasets were larger.
In practice, a system that used  {\name} to distill keys from low-entropy sources would gather more data from the environmental noise source to generate a high-quality key.

We found that the office, mobile, and car datasets have low randomness density and did not pass any NIST test before being processed by {\name} because of their heavy reliance on human activity to generate randomness.
During periods of inactivity---such as a refuelling stop during data collection in the car---result in long runs of ones or zeros in the generated bit sequences, causing the NIST tests to fail.
\textsc{Moonshine} performs well even on low-quality datasets like Car and Mobile.

Fig. \ref{fig:bit_extraction}(b) shows the NIST test pass rate of {\name} running on \textsc{VoltKey} as a function of bit sequence length.
The entropy rate of \textsc{VoltKey} (discussed in \S \ref{sec:entropyrate}) is much higher than the corresponding entropy rate of the car, office, and mobile datasets.
The higher entropy rate of \textsc{VoltKey} gives us more entropy per key bit and better NIST test pass rates for the raw bit stream.
Raw bit sequences generated by \textsc{VoltKey} pass just over half the tests in the NIST suite.
{\name} increases the randomness across all data sets we evaluated.

\begin{figure}
    \centering
    \includegraphics[width=\hsize]{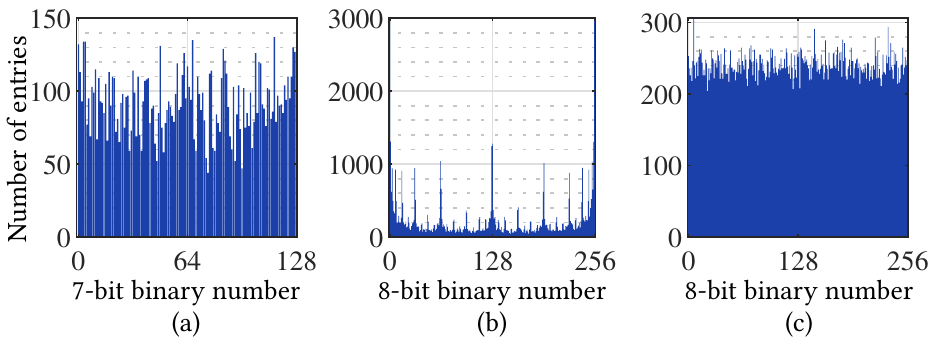}
    \vskip -10pt
    \caption{Histograms of random numbers generated by python's RNG compared to those generated by {\name} and \textsc{VoltKey}. (a) 7-bit sequences after removing elements of nontypical set. (b) Raw 7-bit sequences taken from VoltKey. (c) 7-bit sequences from Python's RNG.}
    \vskip -10pt
    \label{fig:bit_extraction}
\end{figure}

\begin{figure*}[!ht]
    \centering
    \vskip -10pt
    \includegraphics[width=\textwidth]{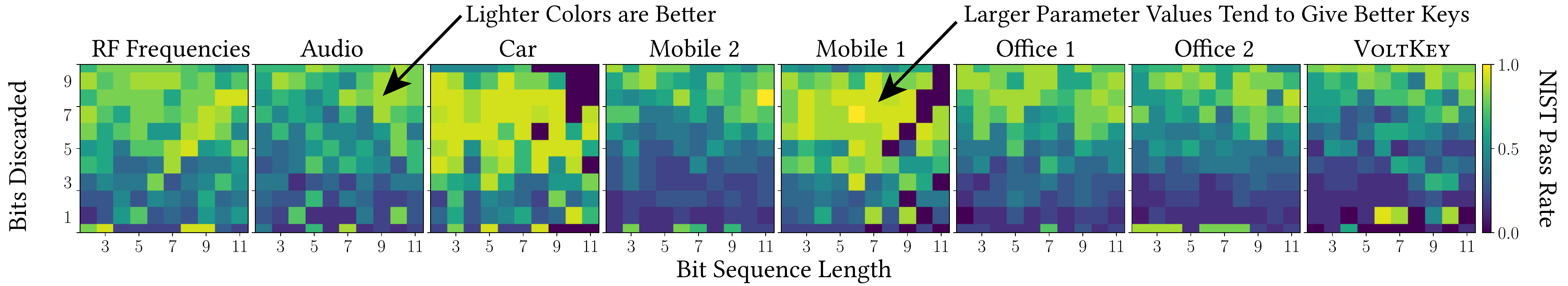}
    \vskip -10pt
    \caption{Fraction of NIST tests passed after applying {\name} as a function of discard length (on the y-axis) and bit sequence length (on the x-axis).}
    \label{fig:discardlength}
\end{figure*}

\begin{figure*}[!ht]
    \centering
    \vskip -10pt
    \includegraphics[width=\textwidth]{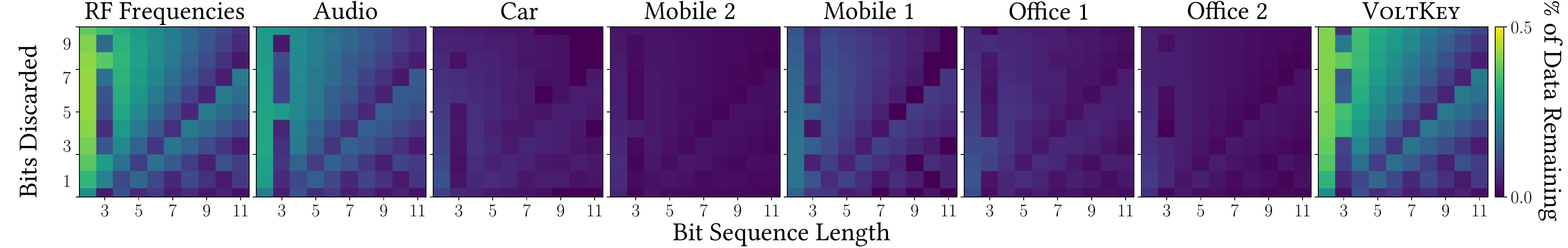}
    \vskip -10pt
    \caption{Fraction of data retained after applying {\name} as a function of discard length (on the y-axis) and bit sequence length (on the x-axis).}
    \label{fig:dataretention}
\end{figure*}

\subsection{Discard Length}

In this section, we test how the NIST test pass rate depends on the number of bits we discard (parameterized by $m$ in Algorithm \ref{algo:max}) when processing the input stream.
We discussed the technique of discarding bits in the input stream in \S \ref{sec:realsignals}.
Fig. \ref{fig:discardlength} shows a heat map of the NIST test pass rate as a function of discard length (y-axis) and input sequence length (x-axis).
Blue colors represent low NIST pass rates, and greens represent higher pass rates (lighter blues and greens are better).

In Fig. \ref{fig:discardlength}, there is a general correlation between data retention and NIST Test pass rate---the more data {\name} removes, the more NIST tests it passes, resulting in higher-quality keys.
Furthermore, the datasets that are higher quality to begin with (such as, \textsc{VoltKey}, Audio, and RF) retain more data after applying {\name}.
The lower-quality datasets have more repetition that {\name}  must remove to produce a high-quality key.

\subsection{Data Retention}

In this section, we evaluate the fraction of input bits that remain in a key after being processed by {\name}.
Fig. \ref{fig:dataretention} is a color graph which represents the percentage of data retained after {\name} distills entropy from the input bit stream.
Data retained is a function of the bit sequence length $k$ and the discard length $m$.

\textsc{VoltKey}, which generates raw bit sequences with relatively high entropy rates, shows an increasing trend of data retained as bit sequence length increases.
According to the AEP, longer sequences of samples will be more uniformly distributed and have a wider separation between the typical set and the nontypical set.
Since {\name} throws away elements of the nontypical set, there will be less information to throw away as sequence lengths get longer.

The office, mobile, and car data exhibit the same trend.
However, the datasets retain less data.
This is caused by the long periods of inactivity in those datasets.
Long runs of ones and zeros on the input bit sequence represent a large percentage of the total bits extracted.
{\name} eliminates those long runs of ones and zeros even when $k$ is small.

{\name} works by selectively dropping bits from the raw input bit stream, so we expect that the output bit sequences will be shorter than the input.
The bits that are dropped by {\name} are those that carry the least amount of information.

\section{Discussion}

In this section, we give some practical advice for designers of context-based authentication mechanisms which is based on the findings in this paper.
We gave two methods of estimating the entropy rate encoded in an environmental noise signal: we calculate the approximate R\'{e}nyi entropy rate of a Gaussian noise process in Theorem \ref{thm:entropyrate}, and we give the more general form of a Shannon entropy rate of a noise process that includes a periodic signal plus additive white Gaussian noise in Equation \ref{eqn:entropyrate}.

\paragraph{Advice \#1: Filter out as much periodic noise as possible from the environmental noise signal before digitizing}

The goal of any context-based key generation scheme is to maximize the mutual information rate between two authenticators.
Since mutual information of two random variables is upper-bounded by each random variable's entropy, we must maximize entropy of the individual source noise processes is to maximize mutual information.

\paragraph{Advice \#2: Apply a Randomness Corrector to the Authentication Key Before Reconciliation}

The raw bit sequences generated from environmental noise tend to have some predictable structure that causes them to fail the NIST test suite (Table \ref{tab:nist}).
This seems to be caused by a low information density, or randomness per bit, in the extracted bit stream.
Using {\name}, we can distill the randomness from a long bit sequence into a shorter bit sequence, yielding a more secure key.
The penalty that we pay when using a randomness corrector is that we must collect $2-3 \times$ as many bits from the environmental noise process.

\paragraph{Advice \#3: Use a High-Entropy Environmental Noise Source if it is Available.}
{\name} is able to distill entropy from a high-entropy source of randomness like \textsc{VoltKey} to generate cryptographic keys that pass the NIST test suite.
We expect that other context-based authentication systems that generate bit sequences that can pass around half of the NIST tests would benefit significantly from {\name}.
{\name} is also able to improve key quality of bit sequences generated by extremely poor entropy sources like the mobile, office, and car datasets we analyzed in \S \ref{sec:eval}.

\section{Conclusion}

In this work, we presented {\name}, a technique to distill randomness from low-entropy sources.
We observed that keys generated from environmental noise sources are predictable largely because the distribution that of the random variable that we use to generate keys is often nonuniformly distributed.
Our methods take advantage of the asymptotic equipartition theorem, which says that long sequences of iid samples from of any random variable will be almost uniformly distributed even if the the distribution of the individual samples is not uniform.
{\name} uses the AEP to identify and remove sequences of samples that are not uniformly distributed.

In our evaluation, we use the NIST test for randomness to evaluate the quality of keys generated by {\name} operating on input data from several publicly available datasets.
Our evaluation showed that {\name} produces keys with substantially higher NIST test pass rates than the raw bit sequences extracted from multiple different environmental sources.

\begin{acks}
This work was supported by the National Science Foundation under award CNS-1845469, CNS-2003129, CNS-1838733, CNS-1719336, CNS-1647152, and CNS-1629833.
\end{acks}

\balance
\bibliographystyle{ACM-Reference-Format}
\bibliography{refs}  

\begin{appendices}

\section{Proof of Theorem \ref{thm:stationarity}}
\label{thm:stationarityproof}

Starting with the simple case of $N=2$:

\begin{align}
\begin{split}\label{eq:1}
    E \left[ D_{t_1} D_{t_2} \right] ={}&  E [ \left( a_1 \cos(2 \pi f t_1 + \Theta)+ a_2\cos(4 \pi f t_1 + \Theta) \right) \times \\
         & \left( a_1 \cos(2 \pi f t_2 + \Theta) + a_2 \cos(4 \pi f t_2 + \Theta) \right) ]
\end{split}\\
\begin{split}\label{eq:expandedproduct}
     ={}& E \left[ a_1^2 \cos(2 \pi f t_1 + \Theta)\cos(2 \pi f t_2 + \Theta) \right]  + \\
        & E \left[ a_1 a_2 \cos(2 \pi f t_1 + \Theta) \cos(4 \pi f t_2 + \Theta) \right]  + \\
        & E \left[ a_1 a_2 \cos(4 \pi f t_1 + \Theta) \cos(2 \pi f t_2 + \Theta) \right]  + \\
        & E \left[ a_2^2   \cos(4 \pi f t_1 + \Theta) \cos(4 \pi f t_2 + \Theta) \right]
\end{split}\\
\begin{split}\label{eq:autocorrresult}
     ={}& \frac{a_1^2}{2} E \left[ \cos(2 \pi f (t_1-t_2) ) \right]  + \\
        & \frac{a_2^2}{2} E \left[ \cos(4 \pi f (t_1 - t_2) ) \right]
\end{split}
\end{align}

\noindent
The cross terms in Eq. \ref{eq:expandedproduct} evaluate to zero because the expectation is the inner product of two orthogonal sinusoids.
The terms involving $\Theta$ in Eq. \ref{eq:expandedproduct} evaluate to zero because we are taking the expectation over $\Theta$, which is assumed to be uniformly distributed on $[-\pi,\pi]$.
This leaves us with an autocorrelation function in Eq. \ref{eq:autocorrresult} that depends only on the difference $(t_1 - t_2)$, which is the criterion for stationary.
This same line of reasoning applies for $N>2$.
This result can also be ported to periodic signals by applying the appropriate limits.
For $N>2$, the same line of reasoning applies because expectations that do not involve $\Theta$ will be taken over triples of orthogonal sinusoids.

\section{Proof of Theorem \ref{thm:renyibound}}
\label{sec:renyisumproof}

\begin{lemma}
\label{lemma:renyisumcondition}

Let $D$ and $Z$ be random variables with alphabets $\mathcal{D}$ and $\mathcal{Z}$ respectively.
Let $X = D + Z$ over alphabet $\mathcal{X}$.
Then:

\begin{equation*}
    R(X|D) = R(Z|D)
\end{equation*}

\end{lemma}
\begin{proof}

\begin{align}
R(X|D) &= \sum_ {d \in\mathcal{D} }P_D(d)R(X|D=d) \\
&= \sum_ {d \in\mathcal{D}}P_D(d)(-\log\sum_{x \in\mathcal{X}}P(X=x|D=d)^2 \\
&= \sum_ {d \in\mathcal{D}}P_D(d)(-\log\sum_{x \in\mathcal{X}}P(Z=x-d|D=d)^2 \\
&= \sum_ {d \in\mathcal{D}}P_D(d)(-\log\sum_{z \in\mathcal{Z}}P(Z=z'|D=d)^2 \\
 &= R(Z|D)
\end{align}
\end{proof}

After observing $D$, the remaining R\'{e}nyi uncertainty in $X$ is due to $Z$.
Because conditioning reduces uncertainty:

\begin{equation*}
    R(X) \geq R(X|D) = R(Z|D) = R(Z)
\end{equation*}

\end{appendices}

\end{document}